\documentclass[10pt,twocolumn,conference]{IEEEtran}
\usepackage{epsfig,rotating,setspace,latexsym,amsmath,epsf,amssymb,bm}
\usepackage{graphicx,color}
\usepackage{amsmath,graphicx}
\usepackage{rotating,latexsym,amsmath,amssymb,bm}
\usepackage{cite}
\usepackage{algpseudocode}
\usepackage{algorithmicx}
\usepackage{algorithm}
\usepackage{xfrac}
\usepackage{amsthm}
\usepackage{subfigure}
\usepackage{graphicx}
\usepackage{subfigure}
\usepackage{algpseudocode}
\usepackage{algorithm}
\usepackage{wrapfig}
\usepackage{dsfont}
\usepackage{wasysym}
\usepackage{bbm}
\usepackage{hyperref}
\usepackage{balance}

\newcommand{\myeq}[1]{\mathrel{\overset{\makebox[0pt]{\mbox{\normalfont\tiny\sffamily #1}}}{=}}}
\newcommand{\myleq}[1]{\mathrel{\overset{\makebox[0pt]{\mbox{\normalfont\tiny\sffamily #1}}}{\leq}}}

\newcommand*{\Scale}[2][4]{\scalebox{#1}{$#2$}}%
\newcommand{\n}{\noindent}
\theoremstyle{plain}
\newtheorem{thm}{Theorem}
\newtheorem{lemma}{Lemma}
\newtheorem{Cor}{Corollary}

\theoremstyle{definition}
\newtheorem{remark}{Remark}
\newtheorem{Def}{Definition}
\theoremstyle{definition}
\newtheorem{example}{Example}

\usepackage{enumitem}
\setlist[enumerate]{leftmargin = 9pt}
\setlist[itemize]{leftmargin = 10pt}
\begin{document}

\setlength\belowdisplayskip{2.5pt}
\setlength\abovedisplayskip{2.5pt}
\allowdisplaybreaks

\sloppy
 
\title{Cache Aided Wireless Networks: Tradeoffs between Storage and Latency}
\author{\IEEEauthorblockN{Avik Sengupta}
\IEEEauthorblockA{Hume Center, Department of ECE\\
Virginia Tech, \\Blacksburg, VA 24060, USA \\
Email: aviksg@vt.edu}
\and
\IEEEauthorblockN{Ravi Tandon}
\IEEEauthorblockA{Department of ECE\\
University of Arizona,\\ Tucson, AZ 85721 USA \\
Email: tandonr@email.arizona.edu}
\and
\IEEEauthorblockN{Osvaldo Simeone}
\IEEEauthorblockA{CWCSPR, Department of ECE\\
New Jersey Institute of Technology, \\Newark, NJ 07102 USA \\
Email: osvaldo.simeone@njit.edu } }

\maketitle 
\thispagestyle{plain}
\pagestyle{plain}
\begin{abstract} 
We investigate the fundamental information theoretic limits of cache-aided wireless networks, in which edge nodes (or transmitters) are endowed with caches that can store popular content, such as multimedia files. This architecture aims to localize popular multimedia content by proactively pushing it closer to the edge of the wireless network, thereby alleviating backhaul load.  An information theoretic model of such networks is presented, that includes the introduction of a new  metric, namely normalized delivery time (NDT), which captures the worst case time to deliver any requested content to the users. We present new results on the trade-off between latency, measured via the NDT, and the cache storage capacity of the edge nodes. In particular, a novel information theoretic lower bound on NDT is presented for cache aided networks. The optimality of this bound is shown for several system parameters.  
\end{abstract}


\section{Introduction}
Edge processing is one of the emerging trends in the evolution of 5G networks \cite{cache_5g}. It refers to the utilization of locally stored content and computing resources at the network edge, i.e., closer to the users. 
Such localization is particularly appealing for both low-latency or location-based applications as well as multimedia transmissions. A network architecture with edge processing capability is shown in Fig. \ref{fig:ITmodel}. Here, edge nodes (ENs), such as base stations or eNodeBs in LTE, are equipped with local caches which can store popular content, most notably multimedia files. The local availability of popular content at the network edge has the potential of reducing the delivery latency as well as the overhead on backhaul connections to content servers. As a result, cache enabled networks have been studied extensively in recent literature \cite{Maddah-Ali, Femto-journal,aviksg-iswcs,aviksg-isit15,aviksg-tifs,Molisch-onecache}.%,aviksg-ITA}. 

In this paper, we investigate \emph{cache-aided wireless networks}, where ENs are endowed with  caching capability to store popular content locally. 
%In contrast to traditional wireless networks, cache-aided wireless networks operate over multiple time scales. 
The design of cache-aided wireless networks involves two key design questions: a) \textit{what to cache}, i.e., how should the storage at ENs be utilized, and which content must be stored; and b) \textit{how to efficiently deliver} the requested content to the users by leveraging the caches at the ENs. The design of \textit{caching policies} is typically done at the long time scale at which users' preferences are invariant and can span many transmission intervals, each corresponding to a set of requests from the users. Hence, the caching policy must be agnostic to the demands of the users as well as to the instantaneous wireless channel conditions. Instead, efficient delivery of requested content to users in each transmission interval calls for the \textit{design of transmission policies} that utilize the available wireless channel state information (CSI) at the ENs and the instantaneous demands of the users. 
%----------------------------------------------------------------------------
\begin{figure}[t]
\centering \hspace{-15pt}
\subfigure[]{
\includegraphics[width=1.7in,height=1.5in]{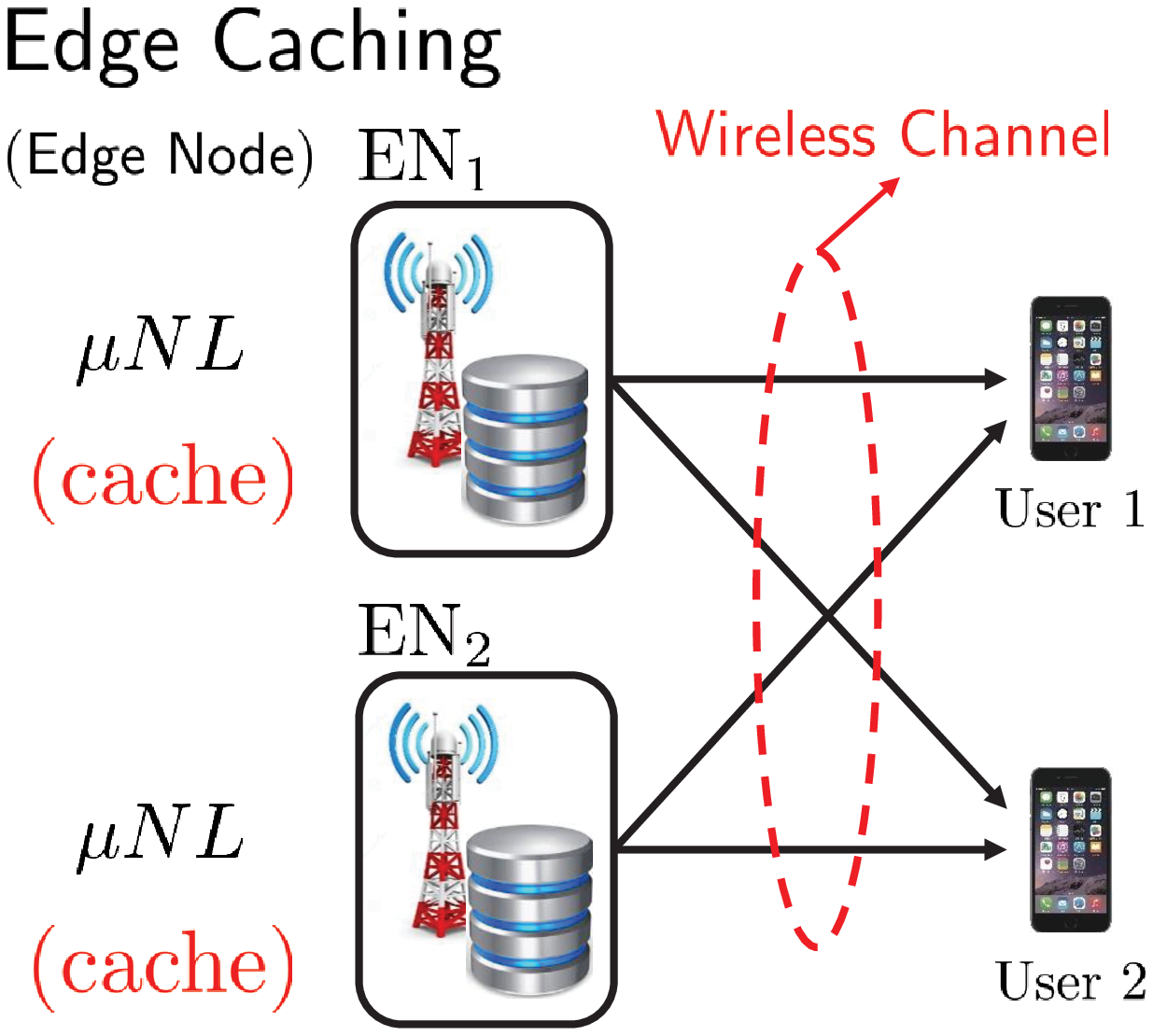}
\label{fig:ITmodel}
}\vspace{-10pt}\\
\subfigure[]{
\includegraphics[width=2.8in,height=1.75in]{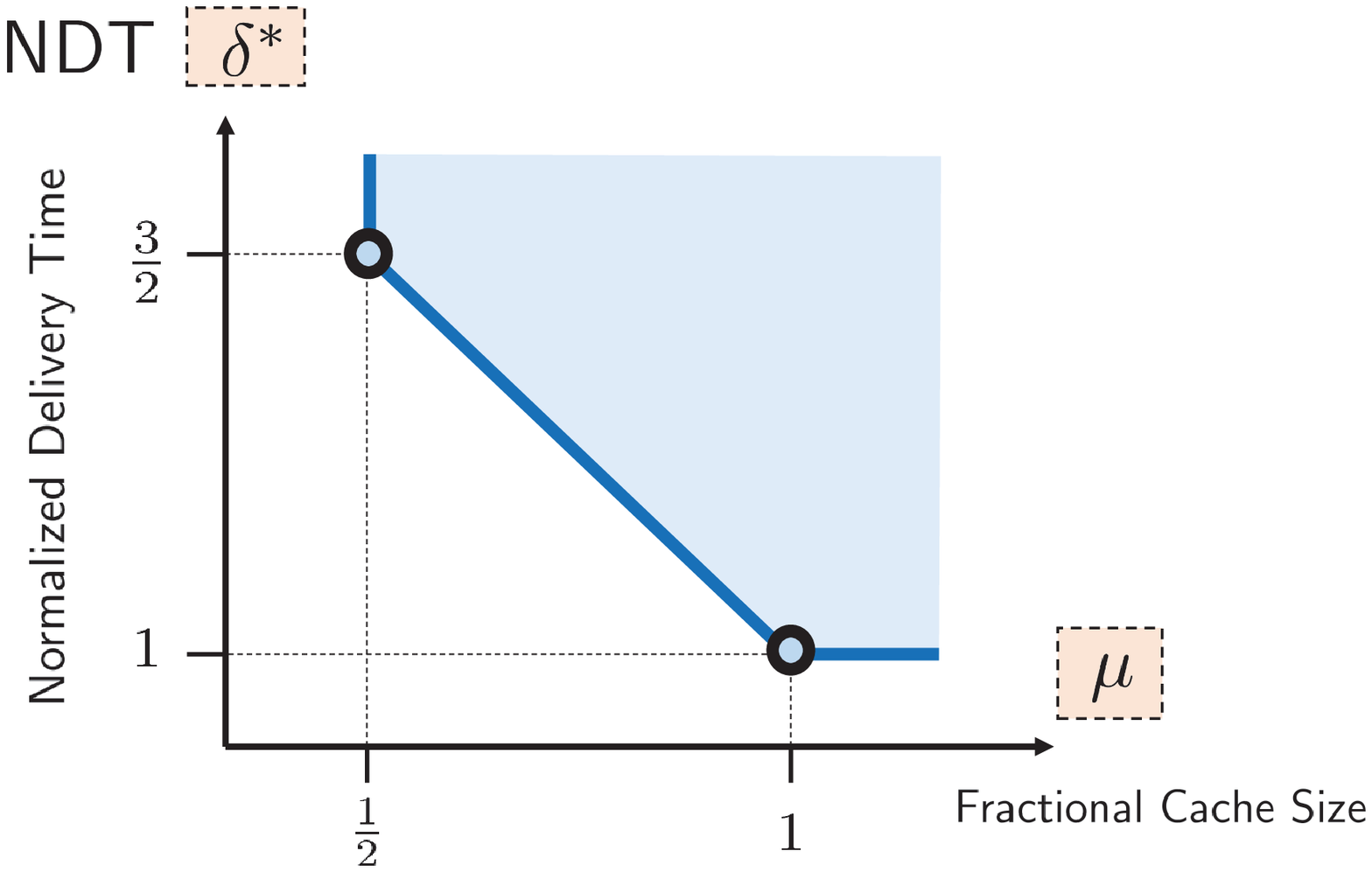}
\label{fig:MK22}
}\vspace{-8pt}
\caption{{(a)} Information-theoretic model for edge caching for $M=2$ ENs serving $K=2$ users; {(b)} Trade-off between the introduce metric of normalized delivery time (NDT), $\delta^{*}$, and the fractional cache size $\mu$ with full CSI at ENs and users.} \vspace{-19pt}
\end{figure}
%----------------------------------------------------------------------------
We first present an information theoretic modeling of cache-aided wireless networks that succinctly captures its new design aspects and constraints. We then develop a new performance measure for such networks termed the \textit{Normalized Delivery Time} (NDT), which measures the worst-case latency incurred by a cache-aided wireless network relative to an ideal system with unlimited caching capability and interference-free links to the users. This facilitates a latency centric analysis of the high signal-to-noise ratio (SNR) degrees-of-freedom (DoF) performance of the system.  

\begin{example}
To illustrate the NDT performance metric, consider the set-up of Fig. \ref{fig:ITmodel}, in which two ENs, labeled $\textrm{EN}_1$ and $\textrm{EN}_2$ are deployed to serve two users over a shared wireless channel. We assume that there is a library of $N$ popular files, each of size $L$ bits, and each EN can cache at most $\mu NL$ bits. In other words, $\mu \in [0,1]$ denotes the \emph{fractional cache size}, i.e., the  ratio between the available per-file storage at an EN and the total size of all the files. For the example shown in Fig. \ref{fig:ITmodel}, the information-theoretically optimal trade-off between NDT $\delta^{*}(\mu)$ and the fractional cache size $\mu$ is shown in Fig. \ref{fig:MK22}. To explain the operating points on this curve, first consider $\mu=1$, i.e., the case when both ENs can store all files, and full cooperative transmission is possible from the ENs, i.e., via zero-forcing beamforming for any set of users' requests. This yields an NDT of $\delta^{*}(1)=1$, implying that the latency performance is the same as that of the ideal interference free system. On the other hand, at $\mu=1/2$, which is the smallest cache for enabling the delivery of any set of requests, the NDT increases to $\delta^{*}(1/2)= 3/2$, and is achieved via interference alignment, thus revealing the performance loss due to decrease in the fractional cache size. \vspace{-5pt}
\end{example} 

\n \textbf{Related Work:} Cache-aided interference channels were first investigated in a recent work by Maddah Ali and Niesen \cite{MA-CAIC,MA-CAIC-arxiv}, who introduced the problem and investigated it for $M=3$ ENs and $K=3$ users, and presented an upper bound on the NDT for this specific setting of $M=K=3$. However, no attempt was made in \cite{MA-CAIC} to develop lower bounds on NDT or to show the optimality of the scheme. 

\n \textbf{Main Contributions:} To the best of our knowledge, this paper is the first to develop information theoretic lower bounds (converse) on latency in cache-aided wireless networks. The main questions we investigate in this work are the following: \emph{What is the optimal caching-transmission policy as a function of the fractional cache size $\mu$? What is the optimal trade-off between the system performance (measured in terms of NDT), and the fractional cache size $\mu$?}
The main contributions of this paper are as follows:
\begin{itemize}
\item We first present an information-theoretic modeling of cache enabled wireless networks and develop the NDT to measure the latency performance of such networks.  For a class of practically relevant caching policies, namely uncoded caching, with full CSI at the ENs, we develop information theoretic lower bounds on the NDT.
\item We show that the presented lower bounds on the NDT are optimal for the setting of $M=2$ ENs and $K=2$ users. Together with the upper bound in \cite{MA-CAIC}, we partially characterize the NDT trade-off for $M=3$ ENs and the $K=3$ users. In addition, we show that our lower bounds are optimal for extremal values of $\mu$ for general problem parameters.
\item Finally, we investigate the impact of CSI availability at the ENs on the NDT. For the case of $M=2$ ENs and $K=2$ users, we illustrate the impact of delayed or no CSI at the ENs on the resulting NDT. 
\end{itemize} \vspace{-5pt}
\section{System Model}\vspace{-5pt}\label{sec:sysmod}
We consider a $M\times K$ cache-aided wireless network where $M$ ENs are connected to a total of $K$ users. The ENs can cache content from a library of $N$ files, $F_1,\ldots,F_N$, where each file is of size $L$ bits, for some $L\in \mathbb{N}^+$. Formally, the files $F_n$ are independent identically distributed (i.i.d.) as:
\begin{align}
F_n \sim \text{Unif}\{1,2,\ldots, 2^{L}\}, ~~\forall n = 1,\ldots,N.
\end{align}
Each EN is equipped with a cache in which it can store $\mu NL$ bits, where the fraction $\mu$, with $0\leq \mu\leq 1$, is referred to as the \textit{fractional cache size}. It is required that the collective cache size of the $M$ ENs be large enough to completely store the entire library of $N$ files. In this way, all user requests can be completely serviced by the ENs. Thus, we impose the condition that $M\times \mu NL \geq N L $, i.e., $ \mu \geq {1}/{M}$. Therefore, it suffices to focus on the range $\mu\in[1/M,1]$. In each transmission interval, a user can request any file from the library and these requests are served by the ENs. The channel between $\textrm{EN}_m$ and user $k$, in a given transmission interval is denoted by $h_{km}\in \mathbb{C}$, where $k = 1,\ldots,K$ and $m = 1,\ldots,M$. The coefficients are assumed to be drawn i.i.d. from a continuous distribution and to be time-invariant with each transmission interval.  

\vspace{-5pt}\begin{Def}[Policy]\label{def:pol}
A caching, edge transmission, decoding policy $\pi = (\pi_c,\pi_e,\pi_d)$ is characterized by the following three functions.

\noindent\textit{a)~\underline{Caching Policy $\pi_c$}:} The caching policy is defined by a function, $\pi_c^{m}(\cdot)$, at each edge node $\textrm{EN}_m, ~m=1,2,\ldots,M$, which maps each file to its cache storage
\begin{align}
S_{m,n} \triangleq \pi_c^{m}\left(F_n\right) ~~\forall n=1,2,\ldots,N.
\end{align}
The mapping is such that $H(S_{m,n})\leq \mu L$ in order to satisfy the cache capacity constraints. The total cache content at $\textrm{EN}_m$ is given by: 
\begin{align}
S_m = \left(S_{m,1},S_{m,2},\ldots,S_{m,N}\right).
\end{align}
Note that the caching policy $\pi_c$ allows for arbitrary coding within each file. However, it does not allow for inter-file coding and is hence a special case of a more general caching policy which might allow for arbitrary inter-file coding. Furthermore, the caching policy is kept fixed over multiple transmission intervals and is thus agnostic to user requests and to channel coefficients $h_{km}$. 

\noindent b) \underline{\textit{Edge Transmission Policy} $\pi_e$}: During the delivery phase of each transmission interval, each receiver $k$ can request one of the $N$ files. We denote by $F_{d_k}$, the file demanded by the $k$-th user, where $d_k \in \{1,2,\ldots,N\}$. The demand vector is denoted by $\mathbf{D}\triangleq \left(d_1,d_2,\ldots,d_K\right) \in \{1,2,\ldots,N\}^K$. Knowing the demand vector $\mathbf{D}$, the global CSI 
\begin{align}
\mathbf{H} = \left\{h_{km}: ~{\substack{k=1,\ldots,K\\ m=1,\ldots,M}}\right\},
\end{align}
denoting the channel coefficient between every user and EN, and having access only to its local cache content, $S_m$, the edge-node $\textrm{EN}_m$ uses an edge transmission policy, $\pi^m_e(\cdot)$, which encodes the cache content, $S_m$, to output a codeword 
\begin{align}
\left(X_m[t]\right)_{t=1}^{T^{(\mathbf{D},\mathbf{H})}} = \pi^m_e\left(S_m,\mathbf{D},\mathbf{H}\right),
\end{align}
which is transmitted over the wireless channel. Here, $T^{(\mathbf{D},\mathbf{H})}$ is the duration or block-length, of the edge transmission policy based on a demand vector $\mathbf{D}$ and the channel realization $\mathbf{H}$. An average power constraint of $P$ is imposed on each codeword, i.e.
\begin{align}\label{eq:powcon}
E\left[\big(X_m[t] - E[X_m[t]]\big)^2 \right] \leq P ~~~\forall t.
\end{align}
We assume that full CSI is available at all ENs and users. The issue of performance losses incurred due to degraded CSI is briefly addressed in Section \ref{ssec:degcsi}.

\noindent c) \underline{\textit{Decoding Policy} $\pi_d$}: Each user $k \in \{1,2,\ldots,K\}$, receives a channel output $\left(Y_k[t]\right)_{t=1}^{T^{(\mathbf{D},\mathbf{H})}}$, given by
\begin{align}
Y_k[t] = \sum_{m=1}^{M}h_{km} X_m[t] + n_k[t] ~~~~\forall t,
\end{align}
where the noise $n_k[t] \sim \mathcal{N}(0,1)$ is a zero mean, unit variance Gaussian random variable which is i.i.d. across time and users. Each user has a decoding policy $\pi_d(\cdot)$, which maps these channel outputs, $\left(Y_k[t]\right)_{t=1}^{T^{(\mathbf{D},\mathbf{H})}}$, the receiver demands $\mathbf{D}$ and the channel realization $\mathbf{H}$ to the estimate
\begin{align}
\widehat{F}_{d_k} \triangleq \pi^k_d\Big(\left(Y_k[t]\right)_{t=1}^{T^{(\mathbf{D},\mathbf{H})}},\mathbf{D},\mathbf{H}\Big)
\end{align}
of the requested file ${F}_{d_k}$. The caching, edge transmission and decoding policies together form a policy $\pi = (\pi_c^{m},\pi_e^m,\pi_d^k)$ for the cache-aided wireless network.
The probability of error of the policy $\pi$ is defined as
\begin{align}
P_e = \max_{\mathbf{D}}\max_{k\in \{1,2,\ldots,K\}} \mathbb{P}\left(\widehat{F}_{d_k} \neq {F}_{d_k}\right).
\end{align}
A policy is said to be feasible if, for almost all realizations $\mathbf{H}$ of the channel, i.e., with probability $1$, we have $P_e \rightarrow 0$ when $L\rightarrow \infty$. 
\end{Def}
\begin{Def}(Delivery time per bit)\label{def:delta}
The \textit{average achievable delivery time per bit} for a given feasible policy is defined as 
\begin{align}
\Delta(\mu,P) &= \max_{\mathbf{D}}\limsup_{L\rightarrow\infty} \frac{\mathbb{E}_{\mathbf{H}}\left[T^{(\mathbf{D},\mathbf{H})}\right]}{L}, 
						%&= \max_{\mathbf{D}\in[N]^K} \lim_{L\rightarrow\infty} \sup\frac{T_{\mu}}{L},
\end{align}
where the expectation is over the channel realizations $\mathbf{H}$. 
%$\Delta(\mu)$ is achievable if there exists a caching policy $\pi = (\pi_c,\pi_e,\pi_d)$,
\end{Def}
%The delivery time per bit accounts for the latency within each transmission interval, evaluated for the worst case user demands and on average over the channel distribution.
While $\Delta(\mu,P)$ generally depends on the power level $P$, as well as on $\mu$, we next define a more tractable metric that reflects the latency performance in the high SNR regime.

\begin{Def}(\textbf{NDT})
For any achievable $\Delta(\mu,P)$, the \textit{normalized delivery time} (NDT), is defined as
\begin{align} \label{eq:ndt1}
\delta(\mu) = \lim_{\substack{P\rightarrow \infty}}\frac{\Delta(\mu,P)}{1/\log P}.
\end{align}
Moreover, for a given $\mu$, the minimum NDT is defined as 
\begin{align}\label{eq:ndt}
\delta^*(\mu) = \inf \left\{\delta(\mu):\delta(\mu) ~\text{is achievable} \right\}.
\end{align}
\end{Def}
\begin{remark}\label{rem:1}
The delivery time per bit $\Delta(\mu,P)$ is normalized by the term $1/\log P$. This is the delivery time per bit in the high SNR regime for an ideal baseline system with no interference and unlimited caching, in which each user can be served by a dedicated EN which has locally stored all the files. An NDT of $\delta^*$ indicates that the worst-case time required to serve any possible request $\mathbf{D}$, is $\delta^*$ times larger than the time needed by this ideal baseline system. 
%For the $M\times K$ edge-caching system defined above, we study the trade-off between the fractional cache size $\mu$ and NDT $\delta^*(\mu)$. 
\end{remark}
\begin{remark}\label{rem:2}
We observe that the NDT in \eqref{eq:ndt} is proportional to the inverse of the more conventional degrees of freedom (DoF) metric $\text{DoF}(\mu)$ defined in \cite{MA-CAIC,MA-CAIC-arxiv}, namely $\delta^{*}(\mu)=K/\text{DoF}(\mu)$. In this paper, we opted for definition \eqref{eq:ndt}, rather than resorting to the DoF metric, as we believe that it more clearly reflects the operational meaning in terms of delivery latency. We also recall that \cite{MA-CAIC,MA-CAIC-arxiv} adopted the metric $1/\text{DoF}(\mu)$ based on the observation that the latter is a convex function of $\mu$, unlike the function $\text{DoF}(\mu)$. Finally, we note that the NDT can be extended to more general scenarios for which a direct functional dependence with the DoF cannot be established \cite{RT_ISIT16}.
\end{remark}

\begin{remark}\label{rem:3}
 Following the same arguments in \cite{MA-CAIC,MA-CAIC-arxiv}, it can be seen that the minimum NDT, $\delta^{*}(\mu)$, is a convex function of $\mu$. In fact, consider any two caching policies $\pi_1$, requiring storage $\mu_1$, and $\pi_2$, requiring storage $\mu_2$. Given a system with storage $\mu = \alpha \mu_1 + (1-\alpha)\mu_2$, for any $\alpha \in [0,1]$, the system can then operate according to policy $\pi_1$ using an $\alpha$-fraction of the cache and of time on the channel to the users, and with policy $\pi_2$ for the remaining part of the cache and of time, achieving an NDT of $\delta^*(\mu)\leq \alpha\delta^*(\mu_1) + (1-\alpha)\delta^*(\mu_2)$. Thus, the convexity of $\delta^*(\mu)$ follows from the possibility of implementing the outlined cache-sharing and time-sharing scheme.
%, as well as time-sharing on the channel to the users between any two policies.   
\end{remark}
\section{Main Results and Discussion}
In this work, we aim to provide fundamental limits for the NDT of an $M\times K$ cache-aided wireless network. To this end, an information theoretic lower-bound on the NDT of the system is presented in the following section under the assumption of perfect CSI at ENs and users. Section \ref{ssec:degcsi}, instead, briefly discusses the impact of imperfect CSI at the ENs.
\subsection{Lower Bounds on NDT with Perfect CSI at the ENs}
In this section, we consider cache-aided wireless networks where perfect CSI is present at the ENs and users. The following Theorem presents an information-theoretic lower bound on the NDT.
\begin{thm}\label{th:1}
For a cache-aided wireless network with $M$ ENs, each with a fractional cache size $\mu \in [1/M,1]$, $K$ users and a library of $N\geq K$ files and with perfect CSI at both ENs and users, the NDT is lower bounded as
\begin{align}\label{eq:th1}
\delta^*(\mu)\geq \max_{\ell \in 1,\ldots,\min\{M,K\}} \frac{K - (M - \ell)^+(K - \ell)^+ \mu}{\ell},
\end{align}
where the function $(x)^+$ is defined as $(x)^+ = \max\{0,x\}.$ 
\end{thm}
To the best of the authors' knowledge, Theorem \ref{th:1} provides the first converse for the $M \times K$ cache-aided wireless network. 
The proof of Theorem \ref{th:1} is presented in Appendix \ref{ap:th1}. To provide further insight into the lower bound in Theorem \ref{th:1}, we present here, a short proof sketch. As shown in Appendix \ref{ap:th1}, the channel outputs of $\ell$ users, along with the cache contents of $(M-\ell)^+$ ENs is sufficient in the high-SNR regime to decode any $K$ requested files. By bounding the joint entropy of these random variables and utilizing the cache storage, caching policy and decodability constraints, one obtains the lower bound on the optimal NDT $\delta^*(\mu)$. Varying the parameter $\ell$ leads to the family of lower bounds in Theorem \ref{th:1}.
 %As expected, the NDT varies inversely with the fractional cache size $\mu$ i.e., a larger cache storage entails a smaller delivery time and vice-versa. 
Based on this lower bound, we next expound on the optimal characterization of $\delta^*(\mu)$ for some cache-aided wireless networks. 
\begin{Cor}\label{cor:corner}
For a cache-aided wireless network with $M$ ENs, each with a fractional cache size $\mu \in [1/M,1]$, $K$ users and a library of $N\geq K$ files, we have
\begin{align}\label{eq:corr11}
\delta^*(\mu) = \frac{M+K-1}{M} ~~~ \text{for}~~ \mu = 1/M,
\end{align}
which can be achieved by leveraging interference alignment techniques for a $M\times K$ X-channel. Further, we have
\begin{align}\label{eq:corr12}
\delta^*(\mu) = \frac{K}{\min\{M,K\}} ~~~ \text{for}~~ \mu = 1,
\end{align}
which can be achieved by using zero-forcing beamforming for a $M\times K$ broadcast channel.
\end{Cor}
\begin{proof} To prove the corollary, we show that a policy with a NDT matching the lower bound in Theorem \ref{th:1} can be identified for both $\mu = 1/M$ and $\mu = 1$.

\noindent \underline{\textit{NDT at} $\mu = 1/M$}: For $\mu = 1/M$, we substitute $\ell = 1$ in \eqref{eq:th1} to get
\begin{align}\label{eq:cp1lb}
\delta^*(1/M) \geq K - \frac{(M-1)(K-1)}{M} = \frac{M+K-1}{M}.
\end{align}
To obtain an upper bound on NDT, we consider the following policy. For $\mu = 1/M$, each file can be split into $M$ non-overlapping fragments $F_n = \left(F_{n,1},F_{n,2},\ldots,F_{n,M}\right)$, each of size $L/M$ bits. The fragment $F_{n,m}$ is stored in the cache of $\textrm{EN}_m$ for $n = 1,\ldots,N$ \cite{MA-CAIC}. Thus, the cache storage for each EN is $NL/M$ bits and the total amount of data stored in the caches of all ENs is $NL$ bits. Next, when a file is requested by any user $k$, each of the ENs have a fragment $F_{d_k,m}$ to transmit to the user. The $M\times K$ system then becomes an X-channel for which, a reliable sum-rate of $\frac{MK}{M+K-1}\log(P)$, neglecting $o(\log(P))$ terms, is achievable by interference alignment \cite{CJ_Xch,MA_Xch}. Thus, the achievable delivery time per bit, in Definition \ref{def:delta}, is approximately given by
\begin{align}
&\hspace{-5pt}\Delta(\mu,P)  = \lim_{L\rightarrow\infty}\frac{1}{L}\cdot\frac{KL}{\frac{MK}{M+K-1}\log(P)} = \frac{M+K-1}{M\log(P)}. \vspace{8pt}
\end{align} 
And hence, we have the achievable NDT
\begin{align}
&\delta(\mu) = \lim_{P\rightarrow \infty} \frac{\Delta(\mu,P)}{1/\log(P)} = \frac{M+K-1}{M}.
\end{align}
Thus, we have the upper bound
\begin{align}\label{eq:cp1ub}
& \delta^*(1/M) \leq \frac{M+K-1}{M}.
\end{align}
Combining \eqref{eq:cp1lb} and \eqref{eq:cp1ub} shows that the lower bound in Theorem \ref{th:1} is tight at $\mu=1/M$.

\noindent \underline{\textit{NDT at} $\mu = 1$}: For $\mu=1$, substituting, $\ell = \min\{M,K\}$ into \eqref{eq:th1}, we get \vspace{-5pt}
\begin{align}\label{eq:cp2lb}
\delta^*(1) \geq \frac{K}{\min\{M,K\}}.
\end{align}
When $\mu=1$, each EN has a cache storage of $NL$ bits, i.e., each EN can completely store the entire library f $N$ files. Hence the ENs can cooperatively transmit to the users using broadcast techniques such as zero-forcing to achieve a reliable sum-rate of $\min\{M,K\}\log(P)$, neglecting $o(\log(P))$ terms \cite{BC_Capacity}. Thus, the delivery time per bit is approximately given by 
\begin{align}
&\Delta(\mu,P) = \lim_{L\rightarrow\infty} \frac{KL/L}{{\min\{M,K\}}\log(P)} =  \frac{K/\log(P)}{{\min\{M,K\}}}. 
\end{align} 
And hence, we have the achievable NDT
\begin{align}
&\delta(\mu) = \lim_{P\rightarrow \infty}\frac{\Delta(\mu,P)}{1/\log(P)} =\frac{K}{\min\{M,K\}}.
\end{align}
%An NDT of $\delta(\mu) = \lim_{P\rightarrow \infty} \Delta(\mu)/(1/\log P)$ is therefore achievable in the high-SNR regime. 
Thus, we have the upper bound
\begin{align}\label{eq:cp2ub}
&~\delta^*(1) \leq \frac{K}{\min\{M,K\}}.
\end{align}
Combining \eqref{eq:cp2lb} and \eqref{eq:cp2ub}, shows that the lower bound in Theorem \ref{th:1} is tight at $\mu=1$.
\end{proof}
\noindent Based on the results of Corollary \ref{cor:corner}, we establish the optimal NDT for a system with $M=K=2$ as stated in the following corollary. 

\begin{Cor}\label{cor:2}
For a cache-aided wireless network with $M = 2$ ENs, $K = 2$ users and $N\geq 2$ files, the optimal NDT is given by 
\begin{align}\label{eq:cor2}
\delta^*(\mu) = 2 -  \mu, ~~~\forall \mu \in [1/2,1]. 
\end{align}
\end{Cor}
For this  $2$-EN, $2$-user system, the two corner points $\mu=1/2$ and $\mu=1$ are achievable as per Corollary \ref{cor:corner}. Instead, for $\mu=1/2$, the system is a $2-$user X-channel which has a sum-DoF of $4/3$, i.e., $\delta(1/2) = 3/2$. Again, at $\mu=1$, the system becomes a broadcast channel which has a sum-DoF of $2$, i.e., $\delta(1) = 1$. All points on the line joining these two achievable points can be achieved through cache and time sharing between the two schemes as stated in Remark \ref{rem:3} in Section \ref{sec:sysmod}. Next, considering the lower bound from Theorem \ref{th:1} and using $\ell=1$, we get \eqref{eq:cor2}, which is the line joining the two achievable corner points. Thus, Theorem \ref{th:1} completely characterizes the optimal NDT $\delta^*(\mu)$ of the cache-aided wireless network with $M=K=2$ and $N\geq 2$. This is illustrated in Fig. \ref{fig:MK22}. 

We next present an application of Theorem \ref{th:1} to obtain a partial characterization of the optimal NDT for a system with $M=3$ ENs and $K=3$ users. 
\begin{Cor}\label{cor:3}
For a cache-aided wireless network with $M = 3$ ENs, $K = 3$ users and $N\geq 3$ files, we have
\begin{align}\label{eq:cor3}
&\delta^*(\mu) = \begin{cases}
								5/3~~~~~~~~~~~~~~~~~~~~~~\text{for} ~~ \mu = 1/3,\\
								3/2 - \mu/2~~~~~~~~~~~~~~\text{for} ~~ 2/3\leq \mu \leq 1\\
								\end{cases} \nonumber\\
& 3 - 4\mu \leq \delta^*(\mu)	\leq 13/6 - 3\mu/2 ~\text{for} ~~ 1/3\leq \mu \leq 2/3.				
\end{align}
\end{Cor}
%----------------------------------------------------------------------------
\begin{figure}[t]
\begin{centering}
\hspace{-2pt}\includegraphics[width=3.5in, height=2.25in]{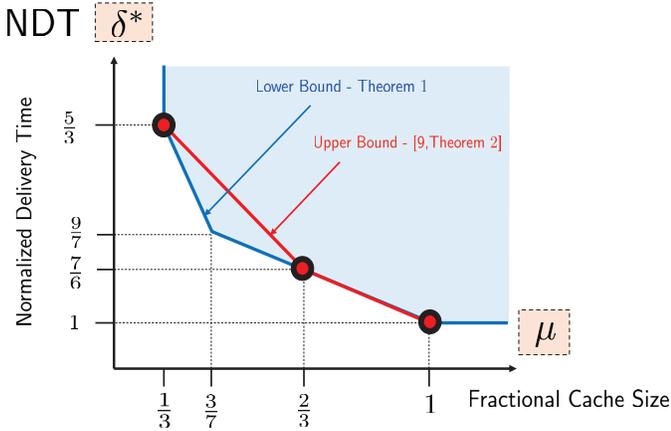}
\vspace{-8.25pt}
\par\end{centering}
\protect\caption{Lower and upper bounds on the NDT for a cache-aided wireless network with $M=3$ ENs and $K=3$ users.}
\label{fig:MK33}\vspace{-5pt}
\end{figure}%
%----------------------------------------------------------------------------

The bounds in Corollary \ref{cor:3} are illustrated in Fig. \ref{fig:MK33}. The lower bounds on NDT used in Corollary \ref{cor:3} are obtained from Theorem \ref{th:1}, by setting $M=K=3$ system, yielding
\begin{align}
\delta^*(\mu)&~\geq~ 3 - 4\mu ~~~~~~~~~\textit{for} ~~~\ell = 1,\\
\delta^*(\mu)&~\geq~ 3/2 - \mu/2 ~~~~~\textit{for} ~~~\ell = 2, \\
\delta^*(\mu)&~\geq~ 1~~~~~~~~~~~~~~~~\textit{for} ~~~\ell = 3.
\end{align}
As for upper bounds, we adapt the results in \cite[Theorem 2]{MA-CAIC} to obtain the following achievable NDT:
\begin{align}\label{eq:MA33}
\delta^*(\mu) \leq \begin{cases}
									13/6 - 3\mu/2~~~~\textit{for}~ 1/3\leq \mu\leq 2/3,\\
									3/2 - \mu/2~~~~~~~\textit{for}~ 2/3\leq \mu\leq 1.
								 \end{cases}
\end{align}
The two corner points for $\mu=1/3$ and $\mu=1$ of the achievable NDT in \eqref{eq:MA33} are achieved similar to Corollary \ref{cor:corner}. The inner point at $\mu = 2/3$, instead uses a novel interference alignment and zero-forcing scheme to achieve a $\delta(\mu) = 7/6$ \cite{MA-CAIC}. It can be seen from Fig. \ref{fig:MK33} that the lower bound coincides with the upper bound at $\mu=1/3$ and for the range $2/3\leq\mu\leq 1$. Hence, the proposed lower bound in conjunction with the recent result from \cite{MA-CAIC}, partially characterizes the optimal NDT versus $\mu$ trade-off for the $M=K=3$ system as summarized in Corollary \ref{cor:3}. For the regime $1/3\leq \mu\leq 2/3$, characterizing the optimal NDT remains an open problem. 
%since it is unclear if memory sharing between the schemes at $\mu=1/3$ and $\mu=2/3$ is the best approach.

\subsection{Impact of Imperfect CSI on the NDT Trade-off}\label{ssec:degcsi}
In this section, we investigate the impact of CSI availability at the ENs and its impact on the NDT. When CSI is delayed, at time $t$, ENs only have access to $\mathbf{H}_1,\mathbf{H}_2,\ldots,\mathbf{H}_{t-1}$, i.e., the CSI of the previous $t-1$ slots. For illustration, we consider the system with $M=K=2$ and $N\geq 2$ with $\mu\in[1/2,1]$. For the case of perfect CSI the optimal NDT can be characterized as in Fig. \ref{fig:MK22}. Next we look at the achievable NDT for the case of delayed and no CSI respectively.

\noindent \textit{a)} \underline{\textit{ Delayed CSI:}} For the case of delayed CSI, consider the corner point $\mu=1/2$ where the system behaves like a $2\times 2$ X-Channel. It is known for the $2\times 2$ X-Channel with delayed CSI that a sum-DoF of $6/5$ is achievable \cite{delayed_csit_Xch}. As a result, an NDT of $\delta(\mu) = 5/3$ is achievable. Compared to the perfect CSI case, the NDT thus incurs a loss due to delayed CSI. Next, consider the corner point $\mu=1$, where the system reduces to a $2\times 2$ broadcast channel with delayed CSI. It is known that for such a system, a sum-DoF of $4/3$ is achievable \cite{delayed_csit_BC}, i.e., a NDT of $\delta(\mu) = 3/2$ is achievable. The optimality of this trade-off is, again, an open problem. However, the achievability illustrates the loss incurred due to delay in CSI availability. 
%-----------------------------------------------------------------------
\begin{figure}[t]
\begin{centering}
\includegraphics[width=2.5in,height=2.3in]{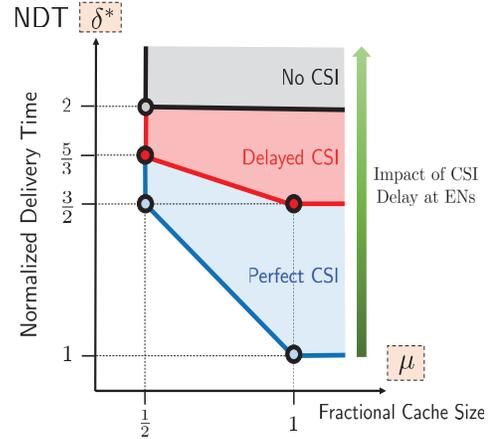}
\par\end{centering}
\protect\caption{ Effect of delayed or no CSI on the NDT for $M=K=2$.}
\label{fig:delayed}
\end{figure}
%-----------------------------------------------------------------------

\noindent \textit{b)} \underline{\textit{No CSI:}} In case of no CSI, it is known that the optimal scheme is transmit using time-division to each user in a separate slot \cite{Vaze_BC_Dof}. Therefore a sum-DoF of $1$ can be achieved, i.e., an NDT of $2$ can be achieved which is optimal for all values of $\mu\in [1/2,1]$. The NDT trade-offs for perfect, delayed and no CSI are shown in Fig. \ref{fig:delayed}.

\section{Conclusions}\label{sec:conc}
In this paper, we studied the fundamental information-theoretic limits of cache-aided wireless networks where network edge nodes are endowed with cache storage. We first proposed an information-theoretic model for such a network and we introduced the metric of normalized delivery time (NDT), which captures the worst-case latency in delivering file requests to users. We presented the first known information theoretic lower bounds for a general $M\times K$ cache-aided wireless networks with perfect CSI. Based on this result, we showed that the optimal NDT for some system parameters can be characterized by the use of known transmission schemes such as interference alignment and zero-forcing beamforming. Finally, we also demonstrated the effect of imperfect (delayed or no) CSI at the ENs and users on the NDT for cache-aided wireless networks.

\appendices
\section{Proof of Theorem \ref{th:1}}\label{ap:th1}
%For the purpose of this proof, we define a block length $T$ for a caching policy $\pi = (\pi_c,\pi_e,\pi_d)$, where $T$ is the lowest time required to service the worst-case requests $\vec{D}$ for any realization of the channel $\mathbf{H}$. We aim to lower bound this delivery time in order to derive a lower bound on the optimal NDT $\delta^*(\mu)$. 

To obtain a lower bound on the NDT, we fix a specific request vector $\mathbf{D}$, namely one for which all requested files $(F_1,...,F_K)=F_{[1:K]}$ are different and a given channel realization $\mathbf{H}$. Note that this is possible given the assumption $N\geq K$. For any integer $a$ and $b$ with $a\leq b$, we define the notation $[a:b] = (a,a+1,\ldots,b)$. We denote as $T$ the delivery  time $T^{(\mathbf{D},\mathbf{H})}$ as per Definition \ref{def:pol} of any given feasible policy $\pi = (\pi_c,\pi_e,\pi_d)$ which guarantees a vanishing probability of error $P_e$ as $L\rightarrow\infty$ for the given request $\mathbf{D}$ and channel $\mathbf{H}$. Our goal is to lower bound $T$ in order to obtain a lower bound on the minimum NDT $\delta^*(\mu)$. 

To this end, consider the channel output $\mathbf{Y}_k^{T} = (Y_k[t])_{t=1}^{T}$ at receiver $k$:
\begin{align}\label{eq:chop}
\mathbf{Y}_k^{T}  =  \sum_{m=1}^M h_{km}\mathbf{X}_m^{T} + \mathbf{n}_k^{T},
\end{align}

\n where $\mathbf{X}_m^{T}  = (X_m[t])_{t=1}^{T}$ and $\mathbf{n}_k^{T} = (n_k[t])_{t=1}^{T}$. We consider $\mathbf{Y}_k^{T}, \mathbf{X}_m^{T}$ and $\mathbf{n}_k^{T}$ as $1\times T$ row vectors. The noise $n_k[t] \sim \mathcal{N}(0,1)$ is a zero mean, unit variance Gaussian random variable and is  i.i.d. across time and users. 

For ease of exposition, we next introduce the following notation which we use throughout the appendix. For any integer pair $(a,b)$ with $1\leq a\leq b\leq K$, let $\mathbf{Y}^T_{[a:b]}$ be the $(b-a+1)\times T$ matrix of channel outputs of a subset $\{a,a+1,\ldots,b\}$, of receivers. The notation is also used for the channel inputs $\mathbf{X}^T$ and noise $\mathbf{n}^{T}$. Furthermore, for any integers $1\leq a\leq b \leq K$ and $1\leq c\leq d\leq M$, we define the following sub-matrix of the channel matrix $\mathbf{H}$:
\begin{align*}  
\mathbf{H}_{[a:b]}^{[c:d]} = \begin{bmatrix} h_{a,c} & h_{a,c+1} & \cdots & h_{a,d} \\
																						 h_{a+1,c} & h_{a+1,c+1} & \cdots & h_{a+1,d} \\
																						 \vdots   & \vdots & \ddots & \vdots \\
																						 h_{b,c} & h_{a,c+1} & \cdots & h_{b,d}
																						\end{bmatrix}.
\end{align*}

\n Using this notation, we can represent the channel outputs at all $K$ receivers as 
\begin{align}\label{eq:chop2}
&\mathbf{Y}_{[1:K]}^{T} = \mathbf{H}_{[1:K]}^{[1:M]}~\mathbf{X}_{[1:M]}^{T} +  \mathbf{n}_{[1:K]}^{T},
\end{align}

To obtain the lower bound on NDT, we make the following key observation, which is illustrated in Fig. \ref{fig:proof}. Given any set of $\ell\leq \min\{M,K\}$ output signals, say $\mathbf{Y}^T_{[1:\ell]}$, and the content of any $(M-\ell)^+$ caches, say $S_{[1:(M-\ell)^+]}$, all transmitted signals $\mathbf{X}^T_{[1:M]}$, and hence also all the files $F_{[1:K]}$, can be resolved in the high-SNR regime. This is because: (\emph{i}) from the cache contents $S_{[1:(M-\ell)^+]}$ one can reconstruct the corresponding inputs $\mathbf{X}^T_{[1:(M-\ell)^+]}$; (\emph{ii}) neglecting the noise in the high-SNR regime, the relationship between the variables $\mathbf{Y}^T_{[1:\ell]}$ and the remaining inputs $\mathbf{X}^T_{[(M-\ell)^+:M]}$ is given almost surely by an invertible linear system as in \eqref{eq:chop}. This intuition is formally stated in Lemma \ref{lem:2} in Appendix \ref{ap:lemma}. We use this argument in the following:
 %Using the notation $[1:x]=\{1,2,\ldots,x\}$, we have
%----------------------------------------------------------------------------
\begin{figure}[t]
\begin{centering}
\includegraphics[width=3.5in, height=2.85in]{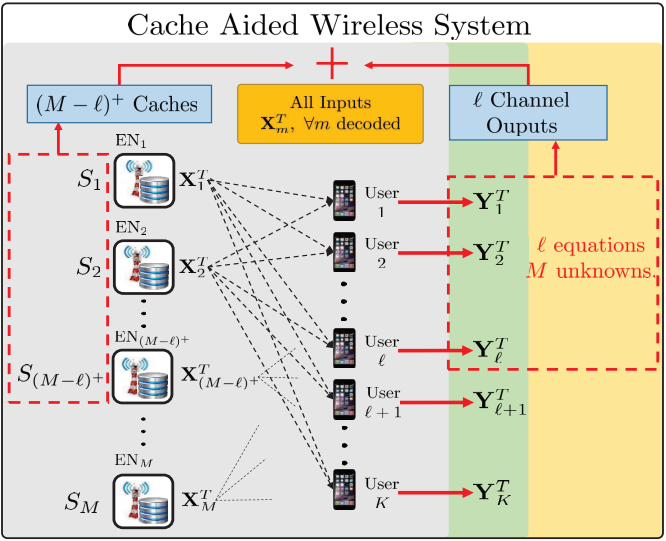}
\par\end{centering}
\protect\caption{Edge-caching set-up for the proof of Theorem \ref{th:1}.}
\label{fig:proof}\vspace{-10pt}
\end{figure}%
\begin{align}\label{eq:p1}
 KL &= H\left(F_{[1:K]}\right) \myeq{(a)} H\left(F_{[1:K]}| F_{[K+1:N]}\right)\nonumber\\
		&= I\left(F_{[1:K]}; \mathbf{Y}^{T}_{[1:\ell]},S_{[1:(M-\ell)^+]}|F_{[K+1:N]}\right) \nonumber\\
		&~~~~+ H\left(F_{[1:K]}|\mathbf{Y}^{T}_{[1:\ell]},S_{[1:(M-\ell)^+]},F_{[K+1:N]} \right),
\end{align}
where $\textsf{(a)}$ follows from the fact that all files $F_{[1:N]}$ are independent. We next upper bound the two terms in \eqref{eq:p1} separately. The first term in \eqref{eq:p1} can be upper bounded as follows: \vspace{5pt}
%\begin{align}\label{eq:A}
 %\mathfrak{A} & = H\left(\mathbf{Y}^{T}_{[1:\ell]},S_{[1:(M-\ell)^+]}|F_{[K+1:N]} \right) \nonumber\\
						 %&~~~~~~~~~~~~~~~~~- H\left(\mathbf{Y}^{T}_{[1:\ell]},S_{[1:(M-\ell)^+]}|F_{[1:N]} \right)\nonumber\\
		%&\myeq{(a)} H\left(\mathbf{Y}^{T}_{[1:\ell]}\right) +  H\left(S_{[1:(M-\ell)^+]} |\mathbf{Y}^{T}_{[1:\ell]},F_{[K+1:N]}  \right) \nonumber\\
		%&- \underbrace{H\Big( S_{[1:(M-\ell)^+]}| F_{[1:N]} \Big)}_{=0} - H\left(\mathbf{Y}^{T}_{[1:\ell]}| S_{[1:(M-\ell)^+]}, F_{[1:N]}\right)\nonumber\\
		%&\myleq{(b)} \ell T\log 2\pi e\left( \tilde{\Lambda} P \right) +\sum_{i=1}^{\Scale[0.6]{(M-\ell)^+}} H\Big(S_{i,[1:N]}|F_{[1:\ell]},F_{[K+1:N]}\Big) \nonumber\\
		%&~~~~~~~ - h\left(\mathbf{Y}^{T}_{[1:\ell]}| F_{[1:N]}\right) \nonumber\\
		%&\myleq{(c)} \ell T\log 2\pi e  \tilde{\Lambda} P  + \sum_{i=1}^{\Scale[0.6]{(M-\ell)^+}}\sum_{j=1}^{\Scale[0.6]{(K-\ell)^+}}H(S_{i,j}) - h\left(\mathbf{n}^T_{[1:\ell]}|F_{[1:N]}\right)\nonumber\\
		%&\myleq{(d)} \ell T\log 2\pi e\left( \tilde{\Lambda} P\right) - \ell T\log 2\pi e  + (M-\ell)^+(K-\ell)^+ \mu L\nonumber\\
		%&\leq \ell T\log \left( \tilde{\Lambda} P\right) + (M-\ell)^+(K-\ell)^+ \mu L,
%\end{align}
\begin{align}\label{eq:A}
\hspace{-25pt} &I\left(F_{[1:K]}; \mathbf{Y}^{T}_{[1:\ell]},S_{[1:(M-\ell)^+]}|F_{[K+1:N]}\right) \nonumber\\
& =I\left(F_{[1:K]}; \mathbf{Y}^{T}_{[1:\ell]} |F_{[K+1:N]}\right)  \nonumber\\
&~~~~~~~+ I\left(F_{[1:K]}; S_{[1:(M-\ell)^+]}|\mathbf{Y}^{T}_{[1:\ell]},F_{[K+1:N]}\right) \nonumber\\
&\leq I\left(F_{[1:K]}; \mathbf{Y}^{T}_{[1:\ell]} |F_{[K+1:N]}\right)  \nonumber\\
&~~~~~~~+ I\left(F_{[1:K]}; S_{[1:(M-\ell)^+]},F_{[1:\ell]}|\mathbf{Y}^{T}_{[1:\ell]},F_{[K+1:N]}\right) \nonumber\\
& = I\left(F_{[1:K]}; \mathbf{Y}^{T}_{[1:\ell]} |F_{[K+1:N]}\right)  \nonumber\\
&~~~~~~~+ I\left(F_{[1:K]}; F_{[1:\ell]}|\mathbf{Y}^{T}_{[1:\ell]},F_{[K+1:N]}\right) \nonumber\\
&~~~~~~~+ I\left(F_{[1:K]}; S_{[1:(M-\ell)^+]}|\mathbf{Y}^{T}_{[1:\ell]},F_{[1:\ell]\cup[K+1:N]}\right) \nonumber\\
&\myleq{(a)} I\left(F_{[1:K]}; \mathbf{Y}^{T}_{[1:\ell]} |F_{[K+1:N]}\right) + H\left( F_{[1:\ell]}|\mathbf{Y}^{T}_{[1:\ell]} \right)\nonumber\\
&~~~~~~~+ H\left(S_{[1:(M-\ell)^+]}|\mathbf{Y}^{T}_{[1:\ell]},F_{[1:\ell]\cup[K+1:N]}\right) \nonumber\\
&~~~~~~~- H\left(S_{[1:(M-\ell)^+]}|\mathbf{Y}^{T}_{[1:\ell]},F_{[1:N]}\right)\nonumber\\
& \myleq{(b)} h\left(\mathbf{Y}^{T}_{[1:\ell]}\right) + L\epsilon_L + H\left(S_{[1:(M-\ell)^+]}|F_{[1:\ell]\cup[K+1:N]}\right) \nonumber\\
& ~~~- H\left(S_{[1:(M-\ell)^+]}|\mathbf{Y}^{T}_{[1:\ell]},F_{[1:N]}\right) - h\left(\mathbf{Y}^{T}_{[1:\ell]}|F_{[1:N]}  \right) \nonumber\\
&\myleq{(c)} \ell T\log\Big(2\pi e \big(\Lambda P + 1\big)\Big)  - h\left(\mathbf{n}^{T}_{[1:\ell]}| F_{[1:N]}\right)  \nonumber\\
&~~~~~~~~~+\sum_{i=1}^{\Scale[0.6]{(M-\ell)^+}} H\Big(S_{i,[1:N]}|F_{[1:\ell]},F_{[K+1:N]}\Big) + L\epsilon_L\nonumber\\
&\myleq{(d)} \ell T\log\Big(2\pi e \big(\Lambda P + 1\big)\Big) - \ell T\log(2\pi e) \nonumber\\
&~~~~~~~~~~~~~~~~~~~~+ \sum_{i=1}^{\Scale[0.6]{(M-\ell)^+}}\sum_{j=1}^{\Scale[0.6]{(K-\ell)^+}}H(S_{i,j}) + L\epsilon_L\nonumber\\
&\leq \ell T\log \big(\Lambda P + 1\big) + (M-\ell)^+(K-\ell)^+ \mu L + L\epsilon_L,
\end{align}\vspace{2.5pt}

%where, $\textsf{(a)}$ follows from the fact that $S_{[1:(M-\ell)^+]}$ are functions of the files $F_{[1:N]}$ and hence the entropy is zero. Step $\textsf{b}$ follows from Lemma \ref{lem:1} stated in Appendix \ref{ap:lemma}. The parameter $\Lambda$ is a constant dependent only on the channel parameters and is defined in Lemma \ref{lem:1}. Also, in $\textsf{(b)}$ it is to be noted that $\mathbf{Y}^T_{[1:\ell]}$ are continuous random variables and hence we consider the differential entropy. Again, since $S_{[1:(M-\ell)^+]}$ are functions of $F_{[1:N]}$, they are dropped from the conditioning. Step $\textsf{(c)}$ follows from the fact that given all the files, the inputs $\mathbf{X}^T_{[1:M]}$ are known and hence can be dropped from the expression in \eqref{eq:chop} leaving only the noise terms. Finally, $\textsf{(d)}$ follows from the fact that the channel noise is i.i.d across time and distributed as $\mathcal{N}(0,1)$. Next, we bound the term $\mathfrak{B}$ as follows
\n where, the steps in \eqref{eq:A} are explained as follows:
\begin{itemize}
\item Step \textsf{(a)} follows from the fact that conditioning reduces entropy.
\item Step \textsf{(b)} follows from the fact that $\mathbf{Y}^{T}_{[1:\ell]}$ are continuous random variables and that dropping the conditioning in the first term increases entropy. We apply Fano's inequality to the second term where $\epsilon_L$ is a function, independent of $P$, which vanishes as $L\rightarrow \infty$.
\item Step \textsf{(c)} can be explained as follows. The first term is upper bounded by the use of Lemma \ref{lem:1} stated below. The parameter $\Lambda$ is a constant dependent only on the channel parameters and is defined in Lemma \ref{lem:1}. The third term is zero since the cache contents $S_{[1:(M-\ell)^+]}$ are functions of the files $F_{[1:N]}$. Moreover, given all the files, the channel outputs are a function of the channel noise at each receiver.
\item Step \textsf{(d)} follows from the fact that the channel noise is i.i.d. across time and distributed as $\mathcal{N}(0,1)$. 
\end{itemize}

\n Next, the second term in \eqref{eq:p1} can be upper bounded as follows:
\begin{align}\label{eq:termB}
&H\left(F_{[1:K]}|\mathbf{Y}^{T}_{[1:\ell]},S_{[1:(M-\ell)^+]},F_{[K+1:N]} \right) \nonumber\\
& ~~~~~~~~~~~~~\leq~~ L\epsilon_L + T\log\det \left( \mathbf{I}_{\Scale[0.7]{[K-\ell]}} + \tilde{\mathbf{H}}\tilde{\mathbf{H}}^H\right),
\end{align}
%The proof of \eqref{eq:termB} is relegated to Lemma \ref{lem:termB} given in Appendix \ref{ap:lemma}. The term $\epsilon_L$ is a function of the probability of error $P_e$ due to Fano's Inequality and 
where $\epsilon_L$ is a function, independent of $P$, that vanishes as $L\rightarrow \infty$; and the matrix $\tilde{\mathbf{H}}$ is square matrix of dimension $(K-\ell)\times (K-\ell)$, which is a function solely of the channel matrix $\mathbf{H}$. The matrix $\mathbf{I}_{\Scale[0.7]{[K-\ell]}}$ is a $(K-\ell)\times (K-\ell)$ identity matrix. We note that the second term in \eqref{eq:termB} is constant independent of the signal power $P$ and the file size $L$. The proof of inequality \eqref{eq:termB} is relegated to Lemma \ref{lem:termB} given in Appendix \ref{ap:lemma}.

 Using \eqref{eq:A} and \eqref{eq:termB} in \eqref{eq:p1}, we get
\begin{align}\label{eq:p2}
& KL \leq \ell T\log \left((\Lambda P + 1)\right) + (M-\ell)^+(K-\ell)^+ \mu L + L\epsilon_L \nonumber\\
&~~~~~~+ T\log\det \left( \mathbf{I}_{\Scale[0.7]{[K-\ell]}} + \tilde{\mathbf{H}}\tilde{\mathbf{H}}^H\right).\nonumber\\
&= \ell T\log(P) \left[1 + \frac{ \ell\log\left(\Lambda + \frac{1}{P}\right) +  \log\det \left( \mathbf{I}_{\Scale[0.7]{[K-\ell]}} + \tilde{\mathbf{H}}\tilde{\mathbf{H}}^H\right) }{\ell\log(P)} \right] \nonumber\\
&~~~~~~~+ (M-\ell)^+(K-\ell)^+ \mu L + L\epsilon_L
\end{align}

%----------------------------------------------------------------------------------------------------------------------
Rearranging \eqref{eq:p2}, we have
\begin{align}\label{eq:p3}
& \frac{T\log(P)}{L}\left[1 + \frac{ \ell\log\left(\Lambda + \frac{1}{P}\right) +  \log\det \left( \mathbf{I}_{\Scale[0.7]{[K-\ell]}} + \tilde{\mathbf{H}}\tilde{\mathbf{H}}^H\right)  }{\ell\log(P)} \right] \nonumber\\
&~~~~~~~\geq \frac{K - (M-\ell)^+(K-\ell)^+ \mu  - \epsilon_L}{\ell}.
\end{align}
Now, using \eqref{eq:p3}, we first take the limit of ${L\rightarrow\infty}$ such that $\epsilon_L \rightarrow 0$  as $P_e \rightarrow 0$. Further, taking the limit ${P\rightarrow \infty}$, for the high-SNR regime, we have
\begin{align}\label{eq:final}
&\delta^*(\mu) \geq  \lim_{\substack{P\rightarrow\infty \\ L\rightarrow\infty}}\frac{T/L}{1/(\log P)} \geq \frac{K - (M-\ell)^+(K-\ell)^+ \mu }{\ell},
\end{align}
%where the term
%\begin{align*}
%\frac{ \ell\log(\Lambda) +  \log \left| \mathbf{I}_{\Scale[0.7]{[K-\ell]}} + \tilde{\mathbf{H}}\tilde{\mathbf{H}}^H\right|  }{\ell\log(P)}
%\end{align*}
in which we have used the fact that the second term within the square brackets in \eqref{eq:p3} vanishes under the limit of $P\rightarrow\infty$. Optimizing the bound in \eqref{eq:final} over all possible choices of $\ell = 1,\ldots,\min\{M,K\}$ completes the proof of Theorem \ref{th:1}.

\section{Lemmas used in Appendix \ref{ap:th1}}\label{ap:lemma}
In this section, we prove the lemmas used in the proof of Theorem \ref{th:1}. First, we state and prove Lemma \ref{lem:1} which was used in \eqref{eq:A} in Appendix \ref{ap:th1}.
\begin{lemma}\label{lem:1}
For the cache-aided wireless network under consideration, the differential entropy of any $\ell$ channel outputs $\mathbf{Y}^{T}_{[1:\ell]}$ can be upper bounded as
\begin{align}\label{eq:lem1}
h\left(\mathbf{Y}^{T}_{[1:\ell]}\right) \leq  \ell{T}\log \Big(2\pi e \left( \Lambda P + 1\right)\Big),
\end{align}
where the parameter $\Lambda$ is a function of the channel coefficients in $\mathbf{H}$ and is defined as 
\begin{align*}
\Lambda = \left(\max\limits_{k\in \{1,\ldots,\ell\}} \left[\sum_{m=1}^M h_{km}^2 + \sum_{m\neq \tilde{m}}h_{km}h_{k\tilde{m}}\right] \right).\vspace{-3pt}
\end{align*}
\end{lemma}
\begin{proof}
The entropy of the received signals $\mathbf{Y}^{T}_{[1:\ell]}$ can be upper bounded as follows:
\begin{align}\label{eq:mi_ub1}
 h\left(\mathbf{Y}_{[1:\ell]}^{T}\right) \leq \sum_{k=1}^\ell \sum_{t=1}^{T} h\Big({Y}_k[t]\Big).
\end{align}
Now, we upper bound the inner sum as follows:
\begin{align}\label{eq:sub1}
&\sum_{t=1}^{T} h\Big({Y}_k[t]\Big) = \sum_{t=1}^{T} h\left( \sum_{m=1}^M h_{km}X_m[t] + n_k[t] \right) \nonumber\\
&\leq \sum_{t=1}^T \log \left(2\pi e ~\text{Var}\left[\sum_{m=1}^M h_{km}X_m[t] + n_k[t] \right]\right)\nonumber\\
& \myeq{(a)} \sum_{t=1}^{T} \log \left(2\pi e \left(\text{Var}\left[\sum_{m=1}^M h_{km}X_m[t]\right] + \text{Var}\left[n_k[t]\right] \right)\right)\nonumber\\
& \myeq{(b)} \sum_{t=1}^{T} \log \Bigg(2\pi e \Bigg(\sum_{m=1}^M h_{km}^2\text{Var}\left[X_m[t]\right] \nonumber\\
& ~~~~~~~~~~~~~~~~~+  \sum_{m\neq \tilde{m}}h_{km}h_{k\tilde{m}}\text{Cov}(X_m[t],X_{\tilde{m}}[t]) + 1 \Bigg) \Bigg)\nonumber\\
& \myleq{(c)} \sum_{t=1}^{T} \log \Bigg(2\pi e \Bigg(\sum_{m=1}^M h_{km}^2\text{Var}\left[X_m[t]\right]  \nonumber\\
&~~~~~~~~~~~+  \sum_{m\neq  \tilde{m}}h_{km}h_{k\tilde{m}}\sqrt{\text{Var}[X_m[t]]\text{Var}[X_{\tilde{m}}[t]]} + 1 \Bigg)\Bigg)\nonumber\\
& \myleq{(d)} \sum_{t=1}^{T} \log \left(2\pi e \left(\sum_{m=1}^M h_{km}^2 P  +  \sum_{m\neq  \tilde{m}}h_{km}h_{k\tilde{m}} P + 1 \right)\right)\nonumber \\
& = \sum_{t=1}^{T} \log \Big(2\pi e \big( {\Lambda} P + 1 \big)\Big)  = T\log \Big(2\pi e \big( {\Lambda} P + 1 \big)\Big) 
%\leq {T}\log \Big(2\pi e\left( \Lambda P \right)\Big),
\end{align} 

\n where ${\Lambda} = \max_{k\in \{1,\ldots,\ell\}}\left[\sum_{m=1}^M h_{km}^2 + \sum_{m\neq  \tilde{m}}h_{km}h_{k\tilde{m}}\right]$. Step \textsf{(a)} in \eqref{eq:sub1} follows from the fact that noise is i.i.d. and uncorrelated with the input symbols; Step \textsf{(b)} follows from the fact that $\text{Var}\left[n_k[t]\right]=1$; Step \textsf{(c)} follows from the Cauchy-Schwartz Inequality; and step \textsf{(d)} follows from (\ref{eq:powcon}). Substituting (\ref{eq:sub1}) into (\ref{eq:mi_ub1}), we have
\begin{align}
h\left(\mathbf{Y}^{T}_{[1:\ell]}\right) 
&~\leq~  \sum_{k=1}^\ell T \log  \Big(2\pi e \big( {\Lambda} P + 1 \big)\Big)  \nonumber\\
&~=~\ell{T}\log  \Big(2\pi e \big( {\Lambda} P + 1 \big)\Big) , \label{eq:mi_ub}
\end{align}
which completes the proof of the Lemma \ref{lem:1}.
\end{proof}

Next, we state and prove Lemma \ref{lem:termB}, which used in \eqref{eq:termB} in the proof of Theorem \ref{th:1}.
\begin{lemma}\label{lem:termB}
For the cache-aided wireless network under consideration, for any feasible policy $\pi = (\pi_c,\pi_e,\pi_d)$, the entropy of the $K$ requested files $F_{[1:K]}$, conditioned on the channel outputs $\mathbf{Y}^{T}_{[1:\ell]}$, on any $(M-\ell)^+$ cache contents $S_{[1:(M-\ell)^+]}$ and on the remaining files $F_{[K+1:M]}$, can be upper bounded as 
\begin{align}
&H\left(F_{[1:K]}|\mathbf{Y}^{T}_{[1:\ell]},S_{[1:(M-\ell)^+]},F_{[K+1:N]} \right) \nonumber\\
& ~~~~~~~~~~~~~\leq~~ L\epsilon_L + T\log\det \left( \mathbf{I}_{\Scale[0.7]{[K-\ell]}} + \tilde{\mathbf{H}}\tilde{\mathbf{H}}^H\right),
\end{align}
where $\epsilon_L$ is a function of the probability of error $P_e$ that vanishes as $L\rightarrow \infty$, the matrix $\tilde{\mathbf{H}}$ is a function solely of the channel matrix $\mathbf{H}$ and $\mathbf{I}_{\Scale[0.7]{[K-\ell]}}$ is a $(K-\ell)\times(K-\ell)$ identity matrix.
\end{lemma}
\begin{proof}
In order to prove this lemma, we first consider the following set of inequalities:
\begin{align}\label{eq:B}
						 &H\left(F_{[1:K]}|\mathbf{Y}^{T}_{[1:\ell]},S_{[1:(M-\ell)^+]},F_{[K+1:N]} \right) \nonumber\\
						 &\myeq{(a)}H\left(F_{[1:K]}|\mathbf{Y}^{T}_{[1:\ell]},S_{[1:(M-\ell)^+]},\mathbf{X}^T_{[1:(M-\ell)^+]},F_{[K+1:N]} \right) \nonumber\\
						 &\myleq{(b)} H\left(F_{[1:K]}|\mathbf{Y}^{T}_{[1:\ell]},\mathbf{X}^T_{[1:(M-\ell)^+]},F_{[K+1:N]} \right)\nonumber\\
						 &\myleq{(c)} H\left(F_{[1:\ell]}|\mathbf{Y}^{T}_{[1:\ell]}\right)	\nonumber\\
						 &~~~~~~~~~+ H\left(F_{[\ell+1:K]}|\mathbf{Y}^{T}_{[1:\ell]},\mathbf{X}^T_{[1:(M-\ell)^+]},F_{[1:\ell]},F_{[K+1:N]} \right)\nonumber\\
						 &\myleq{(d)} L\epsilon_L + H\left(F_{[\ell+1:K]}|\mathbf{Y}^{T}_{[1:\ell]},\mathbf{X}^T_{[1:(M-\ell)^+]},F_{[1:\ell]\cup[K+1:N]} \right),
\end{align}
%%----------------------------------------------------------------------------------------------------
%\begin{figure*}[ht]
%%\small
%\setcounter{equation}{43}
%\begin{align}\label{eq:h1h2}
%&\mathbf{H}_1 = \begin{bmatrix}
								%h_{1,[(M-\ell)^+ + 1]} & h_{1,[(M-\ell)^+ + 2]} & \cdots & h_{1,M}\\
								%h_{2,[(M-\ell)^+ + 1]} & h_{2,[(M-\ell)^+ + 2]} & \cdots & h_{2,M}\\
								%\vdots & \vdots & \ddots & \vdots\\
								%h_{\ell,[(M-\ell)^+ + 1]} & h_{\ell,[(M-\ell)^+ + 2]} & \cdots & h_{\ell,M}\\
							 %\end{bmatrix};~~\mathbf{H}_2 = \begin{bmatrix}
								%h_{\ell+1,[(M-\ell)^+ + 1]} & h_{\ell+1,[(M-\ell)^+ + 2]} & \cdots & h_{\ell+1,M}\\
								%h_{\ell+2,[(M-\ell)^+ + 1]} & h_{\ell+2,[(M-\ell)^+ + 2]} & \cdots & h_{\ell+2,M}\\
								%\vdots & \vdots & \ddots & \vdots\\
								%h_{K,[(M-\ell)^+ + 1]} & h_{K,[(M-\ell)^+ + 2]} & \cdots & h_{K,M}\\
							 %\end{bmatrix}						
%\end{align}
%\hrulefill \setcounter{equation}{31} \vspace{-10pt}
%\end{figure*}
%%----------------------------------------------------------------------------------------------------
\n where the steps in \eqref{eq:B} are explained as follows:
\begin{itemize}
\item Step \textsf{(a)} follows from the fact that the channel inputs $\mathbf{X}^T_{[1:(M-\ell)^+]}$ are functions of the cache contents $S_{[1:(M-\ell)^+}$.
\item Step \textsf{(b)} follows from the fact that the channel inputs $\mathbf{X}^T_{[1:(M-\ell)^+]}$ are functions of the cache contents $S_{[1:(M-\ell)^+]}$.
\item Step \textsf{(c)} follows from the chain rule of entropy and from the fact that conditioning reduces entropy; In step $\textsf{(d)}$, we use Fano's inequality on the first term where $\epsilon_L$ is a function, independent of $P$, that vanishes as $L\rightarrow \infty$. 
\end{itemize}
Next, we consider the second term in \eqref{eq:B}. We have 
\begin{align}\label{eq:c_prel}
& H\left(F_{[\ell+1:K]}|\mathbf{Y}^{T}_{[1:\ell]},\mathbf{X}^T_{[1:(M-\ell)^+]},F_{[1:\ell]\cup[K+1:N]} \right)\nonumber\\
&\myeq{(a)} H\left(F_{[\ell+1:K]}|\mathbf{Y}^{T}_{[1:\ell]},\mathbf{X}^T_{[1:(M-\ell)^+]},\mathbf{n}^T_{[\ell+1:K]},F_{[1:\ell]\cup[K+1:N]} \right)\nonumber\\
&\myleq{(b)} H\Big(F_{[\ell+1:K]}| \mathbf{Y}^T_{[\ell+1:K]} + \tilde{\mathbf{n}}^T_{[\ell+1:K]},\mathbf{Y}^{T}_{[1:\ell]},F_{[1:\ell]\cup[K+1:N]} \Big)\nonumber\\
& \myleq{(c)} H\Big(F_{[\ell+1:K]}| \mathbf{Y}^T_{[\ell+1:K]} + \tilde{\mathbf{n}}^T_{[\ell+1:K]},F_{[1:\ell]\cup[K+1:N]} \Big)\nonumber\\
						&  ~~~- H\Big(F_{[\ell+1:K]}|\mathbf{Y}^T_{[\ell+1:K]},F_{[1:\ell]\cup[K+1:N]} \Big) \nonumber\\
						&  ~~~~~~+  H\Big(F_{[\ell+1:K]}|\mathbf{Y}^T_{[\ell+1:K]},F_{[1:\ell]\cup[K+1:N]} \Big)\nonumber\\
						&\myleq{(d)}  H\Big(F_{[\ell+1:K]}| \mathbf{Y}^T_{[\ell+1:K]} + \tilde{\mathbf{n}}^T_{[\ell+1:K]},F_{[1:\ell]\cup[K+1:N]} \Big)\nonumber\\
						&  ~~~- H\Big(F_{[\ell+1:K]}|\mathbf{Y}^T_{[\ell+1:K]},F_{[1:\ell]\cup[K+1:N]} \Big) + L\epsilon_L 
\end{align}
where the steps in \eqref{eq:c_prel} are explained as follows:
\begin{itemize}
\item Step \textsf{(a)} follows from the fact that the noise term $\mathbf{n}^T_{[\ell+1:K]}$ is independent of all the other random variables in the entropy term and can be introduced into the conditioning. 
\item In Step \textsf{(b)}, we use Lemma \ref{lem:2} stated in Appendix \ref{ap:lemma} and the fact that conditioning reduces entropy. We observe that $\mathbf{n}^T_{[\ell+1:K]} \rightarrow (\mathbf{Y}^{T}_{[1:\ell]},\mathbf{X}^T_{[1:(M-\ell)^+]},F_{[1:\ell]\cup[K+1:N]}) \rightarrow F_{[\ell+1:K]}$ forms a Markov chain and as a result, the data-processing inequality \cite{cover} applies.
%, we observe that $F_{[\ell+1:K]} \rightarrow \left(\mathbf{Y}^{T}_{[1:\ell]},\mathbf{X}^T_{[1:(M-\ell)^+]}, \mathbf{n}^T_{[\ell+1:K]}\right) \rightarrow \left(\mathbf{Y}^T_{[\ell+1:K]} + \tilde{\mathbf{n}}^T_{[\ell+1:K]}\right)$ forms a Markov chain and we apply the data processing inequality \cite{cover}. 
The additive noise term $\tilde{\mathbf{n}}^T_{[\ell+1:K]}$ is defined as  
\begin{align*}
\tilde{\mathbf{n}}^T_{[\ell+1:K]} = \left(\mathbf{H}_2 \cdot {\mathbf{H}_1}^{\dagger}\right)\mathbf{n}^T_{[1:\ell]},
\end{align*}
which is a $[K-\ell]\times T$ matrix, where each column is an independent Gaussian random vector distributed as $\mathcal{N}\left(0,\tilde{\mathbf{H}}\tilde{\mathbf{H}}^H\right)$ with $\tilde{\mathbf{H}} = \left(\mathbf{H}_2 \cdot {\mathbf{H}_1}^{\dagger}\right)$, where the matrices $\mathbf{H}_1$ and $\mathbf{H}_2$ are sub-matrices of the channel matrix $\mathbf{H}$ and are defined in Lemma \ref{lem:2} (see \eqref{eq:h1h2}), and $\mathbf{H_1}^{\dagger}$ is the Moore-Penrose pseudo-inverse. We note here that the noise term $\tilde{\mathbf{n}}^T_{[\ell+1:K]}$ is independent of channel inputs $\mathbf{X}^T_{[1:M]}$ and noise terms $\mathbf{n}^T_{[\ell+1:K]}$.
\item Step \textsf{(c)} follows from the fact that conditioning reduces entropy.
\item Step \textsf{(d)} follows from applying Fano's inequality to the last entropy term in the previous step, where $\epsilon_L$ is again, a function independent of $P$ that vanishes as $L\rightarrow\infty$. 
\end{itemize}
Now, from \eqref{eq:c_prel}, considering the first and second entropy terms together we have:
\begin{align}\label{eq:c12}
%&H\Big(F_{[\ell+1:K]}| \mathbf{Y}^T_{[\ell+1:K]} + \tilde{\mathbf{n}}^T_{[\ell+1:K]},\mathbf{Y}^{T}_{[1:\ell]},F_{[1:\ell]\cup[K+1:N]} \Big)\nonumber\\
%&~~~~~~~~~~ -  H\Big(F_{[\ell+1:K]}|\mathbf{Y}^T_{[\ell+1:K]},F_{[1:\ell]\cup[K+1:N]} \Big) \nonumber\\
& H\Big(F_{[\ell+1:K]}| \mathbf{Y}^T_{[\ell+1:K]} + \tilde{\mathbf{n}}^T_{[\ell+1:K]},F_{[1:\ell]\cup[K+1:N]} \Big)\nonumber\\
&~~~~~~~~~~ -  H\Big(F_{[\ell+1:K]}|\mathbf{Y}^T_{[\ell+1:K]},F_{[1:\ell]\cup[K+1:N]} \Big) \nonumber\\
%& = H\left(F_{[\ell+1:K]}|F_{[1:\ell]\cup[K+1:N]} \right) \nonumber\\
%&~~~~~~~~~~- I\left(F_{[\ell+1:K]}; \mathbf{Y}^T_{[\ell+1:K]} + \tilde{\mathbf{n}}^T_{[\ell+1:K]}|F_{[1:\ell]\cup[K+1:N]} \right)\nonumber\\
%&~~~~~~~~~~- H\left(F_{[\ell+1:K]}|F_{[1:\ell]\cup[K+1:N]} \right) \nonumber\\
%&~~~~~~~~~~+ I\left(F_{[\ell+1:K]};\mathbf{Y}^T_{[\ell+1:K]} |F_{[1:\ell]\cup[K+1:N]} \right)  \nonumber\\
& = I\left(F_{[\ell+1:K]};\mathbf{Y}^T_{[\ell+1:K]} |F_{[1:\ell]\cup[K+1:N]} \right)  \nonumber\\
&~~~~~~~~~~- I\left(F_{[\ell+1:K]}; \mathbf{Y}^T_{[\ell+1:K]} + \tilde{\mathbf{n}}^T_{[\ell+1:K]}|F_{[1:\ell]\cup[K+1:N]} \right)\nonumber\\
& = h\left( \mathbf{Y}^T_{[\ell+1:K]} + \tilde{\mathbf{n}}^T_{[\ell+1:K]}|F_{[1:N]} \right) \nonumber\\
&~~~~~~~~~~- h\left( \mathbf{Y}^T_{[\ell+1:K]} + \tilde{\mathbf{n}}^T_{[\ell+1:K]}|F_{[1:\ell]\cup[K+1:N]} \right)\nonumber\\
&~~~~~~~ + h\left(\mathbf{Y}^T_{[\ell+1:K]} |F_{[1:\ell]\cup[K+1:N]} \right) - h\left(\mathbf{Y}^T_{[\ell+1:K]}|F_{[1:N]} \right)\nonumber\\
& \myleq{(a)} h\left( \mathbf{Y}^T_{[\ell+1:K]} + \tilde{\mathbf{n}}^T_{[\ell+1:K]}|F_{[1:N]} \right) \nonumber\\
&~~~~~~~~~~- h\left( \mathbf{Y}^T_{[\ell+1:K]} + \tilde{\mathbf{n}}^T_{[\ell+1:K]}|\tilde{\mathbf{n}}^T_{[\ell+1:K]},F_{[1:\ell]\cup[K+1:N]} \right)\nonumber\\
&~~~~~~~ + h\left(\mathbf{Y}^T_{[\ell+1:K]} |F_{[1:\ell]\cup[K+1:N]} \right) - h\left(\mathbf{Y}^T_{[\ell+1:K]}|F_{[1:N]} \right)\nonumber\\
& = h\left( \mathbf{Y}^T_{[\ell+1:K]} + \tilde{\mathbf{n}}^T_{[\ell+1:K]}|F_{[1:N]} \right) - h\left(\mathbf{Y}^T_{[\ell+1:K]}|F_{[1:N]} \right)\nonumber\\
%& = H\Big(F_{[\ell+1:K]}| F_{[1:\ell]\cup[K+1:N]} \Big) + H\Big(\mathbf{Y}^T_{[\ell+1:K]} + \tilde{\mathbf{n}}^T_{[\ell+1:K]} | F_{[1:N]} \Big)\nonumber\\
%&~~ - H\Big(F_{[\ell+1:K]}| F_{[1:\ell]\cup[K+1:N]} \Big) -  H\Big(\mathbf{Y}^T_{[\ell+1:K]}|F_{[1:N]} \Big)\nonumber\\
&\myeq{(b)}  h\Big(\mathbf{n}^T_{[\ell+1:K]} + \tilde{\mathbf{n}}^T_{[\ell+1:K]} \Big) - h\Big(\mathbf{n}^T_{[\ell+1:K]}\Big)\nonumber\\
&\myeq{(c)} T\log \Big((2\pi e)^{K-\ell} \left| \mathbf{I}_{\Scale[0.7]{[K-\ell]}} + \tilde{\mathbf{H}}\tilde{\mathbf{H}}^H\right|\Big) - T\log \Big((2\pi e)^{K-\ell}\Big) \nonumber\\
&= T\log\det \left( \mathbf{I}_{\Scale[0.7]{[K-\ell]}} + \tilde{\mathbf{H}}\tilde{\mathbf{H}}^H\right).
\end{align}
%where  and $ \mathbf{I}_{\Scale[0.7]{[K-\ell]}}$ is an identity matrix of dimension $[K-\ell]$.  As noted before, noise terms $\tilde{\mathbf{n}}^T_{[\ell+1:K]}$, are a function of the noise terms $\mathbf{n}^T_{[1:\ell]}$ and the channel coefficients and is hence uncorrelated to the noise terms $\mathbf{n}^T_{[\ell+1:K]}$. 
\n The steps in \eqref{eq:c12} are explained as follows:
\begin{itemize}
\item Step \textsf{(a)} follows from the fact that conditioning reduces entropy.
\item Step \textsf{(b)} follows from the fact that, given all the files $F_{[1:N]}$, the channel outputs are functions of the channel noise.
\item Step \textsf{(c)} follows from the fact that the noise terms are jointly Gaussian and are i.i.d. across time $T$. The function $|\cdot|$ is the determinant.  
\end{itemize}
%First, observe that since $\mathbf{n}^T_{[\ell+1:K]}$ are i.i.d., they are jointly distributed as $\mathcal{N}(\mathbf{0},\mathbf{I}_{\Scale[0.7]{[K-\ell]}})$. Further, based on the definition of $ \tilde{\mathbf{n}}^T_{[\ell+1:K]}$ in Lemma \ref{lem:2}, it can be shown that $\mathbf{n}^T_{[\ell+1:K]} + \tilde{\mathbf{n}}^T_{[\ell+1:K]}$ is jointly distributed as $\mathcal{N}\left(\mathbf{0}, \mathbf{I}_{\Scale[0.7]{[K-\ell]}} + \tilde{\mathbf{H}}\cdot\mathbf{I}_{\Scale[0.7]{[K-\ell]}}\cdot\tilde{\mathbf{H}}^H \right)$, i.e., as $\mathcal{N}\left(\mathbf{0}, \mathbf{I}_{\Scale[0.7]{[K-\ell]}} + \tilde{\mathbf{H}}\tilde{\mathbf{H}}^H \right)$.
%Using the differential entropy of jointly Gaussian random variables leads to the equality. 
%Next, considering the second and fourth entropy terms in \eqref{eq:c_prel}, we have
%\begin{align}\label{eq:c24}
%&H\Big(\mathbf{Y}^T_{[\ell+1:K]}|F_{[1:\ell]\cup[K+1:N]} \Big)\nonumber\\
%&~~~~- H\Big(\mathbf{Y}^T_{[\ell+1:K]} + \tilde{\mathbf{n}}^T_{[\ell+1:K]} | F_{[1:\ell]\cup[K+1:N]}\Big)  \nonumber\\
%&\myleq{(a)}  H\Big(\mathbf{Y}^T_{[\ell+1:K]}|F_{[1:\ell]\cup[K+1:N]} \Big) \nonumber\\
%&~~~~~~ -  H\Big(\mathbf{Y}^T_{[\ell+1:K]} + \tilde{\mathbf{n}}^T_{[\ell+1:K]} |\tilde{\mathbf{n}}^T_{[\ell+1:K]}, F_{[1:\ell]\cup[K+1:N]}\Big) \nonumber\\
%&= H\Big(\mathbf{Y}^T_{[\ell+1:K]}|F_{[1:\ell]\cup[K+1:N]} \Big) - H\Big(\mathbf{Y}^T_{[\ell+1:K]} | F_{[1:\ell]\cup[K+1:N]}\Big) \nonumber\\
%&=0,
%\end{align}
%where, \textsf{(a)} follows from the fact that conditioning reduces entropy. 
Thus, using \eqref{eq:c_prel} and \eqref{eq:c12} in \eqref{eq:B}, we have
\begin{align}\label{eq:C}
&H\left(F_{[1:K]}|\mathbf{Y}^{T}_{[1:\ell]},S_{[1:(M-\ell)^+]},F_{[K+1:N]} \right) \nonumber\\
& ~~~~~~~~~~~~~\leq~~ L\epsilon_L + T\log\det \left( \mathbf{I}_{\Scale[0.7]{[K-\ell]}} + \tilde{\mathbf{H}}\tilde{\mathbf{H}}^H\right),
\end{align}
%Note that Theorem \ref{th:1} holds for almost all channel realizations, i.e., with probability $1$.
which completes the proof of the Lemma \ref{lem:termB}.
\end{proof} 
Finally, we state and prove Lemma \ref{lem:2} which was used in \eqref{eq:c_prel} for the proof of Lemma \ref{lem:termB}.
\begin{lemma}\label{lem:2}
Given any $\ell \in 1,2,\ldots,\min\{N,K\}$, there exists a (deterministic) function of the channel outputs $\mathbf{Y}^T_{[1:\ell]}$, input symbols $\mathbf{X}^T_{[1:(M-\ell)]^+}$ and channel noise $\mathbf{n}^T_{[\ell+1:K]}$, that yields
%a noisy estimate of the channel outputs $\mathbf{Y}^T_{[\ell+1:K]}$ can be obtained as
\begin{align}\label{eq:lem2}
\mathbf{Y}^T_{[\ell+1:K]} + \tilde{\mathbf{n}}^T_{[\ell+1:K]}, 
%= f\left(\mathbf{Y}^{T}_{[1:\ell]},\mathbf{X}^T_{[1:(M-\ell)^+]}\right) + \mathbf{n}_{[\ell+1:K]}^{T} ,
\end{align}
where we have defined $\tilde{\mathbf{n}}^T_{[\ell+1:K]} = \left(\mathbf{H}_2 \cdot {\mathbf{H}_1}^{\dagger}\right)\mathbf{n}^T_{[1:\ell]}$ and $\mathbf{H_1}^{\dagger}$ is the Moore-Penrose pseudo-inverse. The matrices $\mathbf{H}_1$ and $\mathbf{H}_2$ are sub-matrices of the channel matrix $\mathbf{H}$ and are defined as 
\begin{align}\label{eq:h1h2}
&\mathbf{H}_1 = \mathbf{H}^{[1:\ell]}_{[(M-\ell)^+ +1:M]};~~~ \mathbf{H}_2 = \mathbf{H}^{[\ell+1:K]}_{[(M-\ell)^+ +1:M]}.
\end{align}
\end{lemma}
\begin{proof}
Given any $\ell \in 1,2,\ldots,\min\{M,K\}$, from \eqref{eq:chop2}, the channel outputs $\mathbf{Y}^T_{[1:\ell]}$ are a function of the $M$ input symbols $\mathbf{X}^T_{[1:M]}$ and of the noise $\mathbf{n}^T_{[1:\ell]}$. Given the input symbols $\mathbf{X}^T_{[1:(M-\ell)^+]}$, we can cancel the contribution of these input symbols from the channel outputs $\mathbf{Y}^T_{[1:\ell]}$ to obtain
\setcounter{equation}{44}
\begin{align}\label{eq:ty}
\tilde{\mathbf{Y}}^T_{[1:\ell]} &= \mathbf{H}^{[1:M]}_{[1:\ell]}
 \mathbf{X}_{[1:M]}^{T}  +  \mathbf{n}_{[1:\ell]}^{T}  - \mathbf{H}^{[1:M]}_{[1:\ell]}
\begin{bmatrix} \mathbf{X}_{[1:(M-\ell)^+]}^{T} \\ \mathbf{0}^T_{[(M-\ell)^+ +1:M]} \end{bmatrix}\nonumber\\
&  = \mathbf{H}_1 \begin{bmatrix} \mathbf{X}_{[(M-\ell)^+ + 1:M]}^{T} \end{bmatrix} + \begin{bmatrix} \mathbf{n}_{[1:\ell]}^{T}\end{bmatrix},
\end{align}
where $\mathbf{0}^T_{[(M-\ell)^+ +1:M]}$ is an $\ell \times T$ matrix of zeros. As a result, multiplying both sides of \eqref{eq:ty} by $\mathbf{H_1}^{\dagger}$, we get
\begin{align}\label{eq:X}
\mathbf{H_1}^{\dagger}\tilde{\mathbf{Y}}^T_{[1:\ell]}  =  \mathbf{X}_{[(M-\ell)^+ + 1:M]}^{T} + \mathbf{H_1}^{\dagger}\mathbf{n}_{[1:\ell]}^{T}.
\end{align}
Now let
\begin{align}
\mathbf{H}_3 = \mathbf{H}^{[1:M]}_{[\ell+1:K]}.
%\begin{bmatrix}h_{\ell+1,1} & \cdots & h_{\ell+1,M} \\ h_{\ell+2,1}  & \cdots & h_{\ell+2,M} \\ \vdots & \ddots& \vdots\\ h_{K,1} &  \cdots & h_{K,M}\end{bmatrix}.
\end{align}
Using this definition, we have
\begin{align}\label{eq:last}
&\mathbf{Y}_{[\ell + 1:K]}^{T}
= \mathbf{H}_3 \mathbf{X}_{[1:M]}^{T} + \mathbf{n}_{[\ell+1:K]}^{T}\nonumber\\
& = \mathbf{H}_3
\begin{bmatrix}\mathbf{X}_{[1:(M-\ell)^+]}^{T}\\ \mathbf{H_1}^{\dagger}\tilde{\mathbf{Y}}^T_{[1:\ell]} - \mathbf{H_1}^{\dagger}\mathbf{n}^{T}_{[1:\ell]} \end{bmatrix} 
+ \mathbf{n}_{[\ell+1:K]}^{T} \nonumber\\
& \myeq{(a)} \mathbf{H}_3 \begin{bmatrix}\mathbf{X}_{[1:(M-\ell)^+]}^{T}\\ \mathbf{H_1}^{\dagger}\tilde{\mathbf{Y}}^T_{[1:\ell]} \end{bmatrix} 
-  \mathbf{H}_3 \begin{bmatrix} \mathbf{0}^T_{[1:(M-\ell)^+]}\\ \mathbf{H_1}^{\dagger}\mathbf{n}^{T}_{[1:\ell]}\end{bmatrix} 
+ \mathbf{n}_{[\ell+1:K]}^{T} \nonumber\\
& = \mathbf{H}_3 \begin{bmatrix}\mathbf{X}_{[1:(M-\ell)^+]}^{T}\\ \mathbf{H_1}^{\dagger}\tilde{\mathbf{Y}}^T_{[1:\ell]} \end{bmatrix} 
-  \mathbf{H}_2 \begin{bmatrix} \mathbf{H_1}^{\dagger}\mathbf{n}^{T}_{[1:\ell]}\end{bmatrix} 
+ \mathbf{n}_{[\ell+1:K]}^{T},
\end{align}
where, in \textsf{(a)}, $\mathbf{0}^T_{[1:(M-\ell)^+]}$ is a $[(M-\ell)^+]\times T$ matrix of zeros. Rearranging \eqref{eq:last}, we obtain
\begin{align}
\mathbf{Y}^T_{[\ell+1:K]} + \tilde{\mathbf{n}}^T_{[\ell+1:K]}  &=  \mathbf{H}_3 \begin{bmatrix}\mathbf{X}_{[1:(M-\ell)^+]}^{T}\\ \mathbf{H_1}^{\dagger}\tilde{\mathbf{Y}}^T_{[1:\ell]} \end{bmatrix} + \mathbf{n}_{[\ell+1:K]}^{T},
%&\hspace{-70pt}= f\left(\mathbf{Y}^{T}_{[1:\ell]},\mathbf{X}^T_{[1:(M-\ell)^+]}\right) +\mathbf{n}_{[\ell+1:K]}^{T} , 
\end{align}
%which is a noisy estimate of the channel outputs $\mathbf{Y}^T_{[\ell+1:K]}$ where the additional noise terms are functions of the noise at receivers $1$ to $\ell$ and the channel coefficients.
%where the function $f\left(\mathbf{Y}^{T}_{[1:\ell]},\mathbf{X}^T_{[1:(M-\ell)^+]}\right)$ is given by
%\begin{align}\label{eq:func}
%f\left(\mathbf{Y}^{T}_{[1:\ell]},\mathbf{X}^T_{[1:(M-\ell)^+]}\right) = \mathbf{H}_3 \begin{bmatrix}\mathbf{X}_{[1:(M-\ell)^+]}^{T}\\ \mathbf{H_1}^{\dagger}\tilde{\mathbf{Y}}^T_{[1:\ell]} \end{bmatrix}.
%\end{align}
where the RHS is a function of the $\ell$ channel outputs $\mathbf{Y}^T_{[1:\ell]}$, input symbols $\mathbf{X}^T_{[1:(M-\ell)]^+}$ and channel noise $\mathbf{n}^T_{[\ell+1:K]}$.
This completes the proof Lemma \ref{lem:2}. Note that we assumed in \eqref{eq:X} that the sub-matrix $\mathbf{H}_1$ is invertible, which is true for almost all channel realizations, i.e., it is true with probability $1$.
\end{proof}

\balance
\bibliographystyle{IEEEtran}
\bibliography{CISS_arXiv_Final}

\end{document}